\documentclass[fleqn,reqno]{article}
\usepackage{amsmath,amssymb,amsthm}
\usepackage[margin=1.0in]{geometry}
\usepackage{enumerate}

\theoremstyle{plain}
\newtheorem{theorem}{Theorem}
\newtheorem{proposition}[theorem]{Proposition}
\newtheorem{lemma}[theorem]{Lemma}

\theoremstyle{definition}

\newcommand{\R}{\mathbb{R}}
\newcommand{\C}{\mathbb{C}}

\newcommand{\cF}{\mathcal{F}}

\newcommand{\cL}{\mathcal{L}}         
\newcommand{\cP}{\mathcal{P}}         

\newcommand{\cS}{\mathcal{S}}

\newcommand{\E}{\mathcal{E}}

\newcommand{\LAT}{\mathcal{L}}

\newcommand{\Per}{\mathcal{P}}

\newcommand{\nc}{\newcommand}

\nc{\be}{\begin{equation}}
\nc{\la}{\label}
\nc{\ba}{\begin{array}}
\nc{\ea}{\end{array}}
\nc{\bs}{\begin{split}}
\nc{\es}{\end{split}}

\nc{\zb}{\underbar{z}}
\nc{\pb}{\underbar{p}}
\nc{\bb}{\underbar{b}}
\nc{\zbt}{\underbar{z}(t)}

\nc{\dA}{\nabla_A}

\nc{\nj}{(n_j)}
\nc{\nk}{(n_k)}
\nc{\nl}{(n_l)}
\nc{\nr}{(n_r)}
\nc{\nt}{(n_t)}

\nc{\lam}{\lambda}
\nc{\G}{\Gamma}
\nc{\g}{\gamma}
\nc{\al}{\alpha}
\nc{\del}{\delta}
\nc{\Om}{\Omega}
\nc{\Omt}{\tilde{\Omega}}
\nc{\ta}{\tau}
\nc{\w}{\omega}
\nc{\io}{\iota}
\nc{\h}{\theta}
\nc{\z}{\zeta}
\nc{\s}{\sigma}
\nc{\Si}{\Sigma}
\nc{\Lam}{\Lambda}

\nc{\bP}{\bar{P}}
\nc{\bQ}{\bar{Q}}
\nc{\bL}{\bar{L}}

\nc{\ran}{\rangle}
\nc{\lan}{\langle}

\nc{\bfone}{{\bf 1}}

\newcommand{\one}{\mathbf{1}}

\newcommand{\NULL}{\operatorname{null}}

\newcommand{\Ran}{\operatorname{Ran}}
\newcommand{\p}{\partial}

\newcommand{\e}{\epsilon}
\newcommand{\ra}{\rightarrow}
\newcommand{\diag}{\operatorname{diag}}
\renewcommand{\Re}{\operatorname{Re}}
\renewcommand{\Im}{\operatorname{Im}}

\newcommand{\CURL}{\operatorname{curl}}
\newcommand{\divv}{\operatorname{div}}

\newcommand{\DIV}{\operatorname{div}}
\newcommand{\n}{\nabla}

\newcommand{\COVGRAD}[1]{\nabla_{\!\!#1}}
\newcommand{\COVLAP}[1]{\Delta_{\!#1}}

\newcommand{\const}{\operatorname{const}}

\newcommand{\LA}[2]{\vec{\mathscr{L}}_{}^{}(\tau)}
\newcommand{\HA}[2]{\vec{\mathscr{H}}_{}^{}(\tau)}

\newcommand{\DETAILS}[1]{}

\begin{document}


\title
{Stability of Abrikosov lattices under gauge-periodic perturbations}

\author{Israel Michael Sigal
\thanks{ Dept.~of Math.,
Univ. of Toronto, Toronto, Canada; Supported by NSERC Grant No. NA7901}\ \qquad \qquad
Tim Tzaneteas
\thanks{Dept.~of Math.,
Univ. of Toronto, Toronto, Canada; Supported by NSERC Grant No. NA7901} \\ 
}

\date{}
\maketitle

\begin{abstract}
  We consider Abrikosov-type vortex lattice solutions of the Ginzburg-Landau equations of superconductivity, consisting of single vortices, for magnetic fields below but close to the second critical magnetic field $H_{c2}=\kappa^2$ and for superconductors filling the entire $\R^2$. Here $\kappa$ is the Ginzburg-Landau parameter.  The lattice shape, parameterized by $\tau$, is allowed to be arbitrary (not just triangular or rectangular).
 Within the context of the time-dependent Ginzburg-Landau equations, called the Gorkov-Eliashberg-Schmidt equations, we  prove that such lattices are asymptotically stable under gauge periodic perturbations for  $\kappa^2 > \frac{1}{2}(1 - \frac{1}{\beta(\tau)})$  and unstable for $\kappa^2 < \frac{1}{2}(1 - \frac{1}{\beta(\tau)})$, where $\beta(\tau)$ is the Abrikosov constant depending on the lattice shape $\tau$.  This result goes against the common belief among physicists and mathematicians that  Abrikosov-type vortex lattice solutions are stable only for triangular lattices and $\kappa^2 > \frac{1}{2}$. (There is no real contradiction though as we consider very special perturbations.)
\end{abstract}



\section{Introduction}\label{sec:introduction}

\subsection{Background}\label{subsec:backgr}

Macroscopic theory of superconductivity, by now a classical theory presented in any book on superconductivity and solid state or condensed matter physics (see e.g. \cite{Tink, Kit}),  was developed by Ginzburg and Landau along the lines of the Landau theory of the second order phase transitions and before the microscopic theory was discoverd.  At the foundation of this theory lie  the Ginzburg-Landau equations for the order parameter and magnetic potential. The time-dependent generalization of these equations was proposed by Schmidt (\cite{Sch}) and Gorkov and Eliashberg (\cite{GE}) and are known as  the Gorkov-Eliashberg or  Gorkov-Eliashberg-Schmidt equations (as well as  the time-dependent Ginzburg-Landau equations).  The latter equations have much narrower range of applicability than  the Ginzburg-Landau equations (\cite{Tink}) and many refinements of them were proposed, but even a slight improvement of these equations is, at least notationally, extremely cumbersome.

By far, the most important and celebrated solutions of  the Ginzburg-Landau equations of superconductivity are vortices and vortex lattices, discovered by Abrikosov (\cite{Abr}), and  known as Abrikosov (vortex) lattice solutions. Among other things understanding these solutions is important for maintaining the superconducting current in Type II superconductors, i.e. for $\kappa> \frac{1}{\sqrt{2}}$, in our units.

Abrikosov lattice solutions are extensively studied in physics literature. Among many rigorous results, we mention that the existence of these solutions was proven rigorously in  \cite{Odeh, BGT, Dut, Al, as, TS} (\cite{Odeh, BGT, Al} deal with triangular and rectangular lattice, while \cite{Dut, as, TS}, with lattices of arbitrary shape). Moreover,  important and fairly detailed results on asymptotic behaviour of solutions, for $\kappa \to\infty$ and the applied magnetic fields, $h$, satisfying $h\le \frac{1}{2}\log \kappa+\const$ (the London limit), were obtained in \cite{as} (see this paper and the book \cite{ss} for references to earlier works). Further extensions to the Ginzburg-Landau equations for anisotropic and high temperature superconductors can be found in \cite{abs1, abs2}.

  In this paper we address the problem of stability of the  Abrikosov vortex lattice solutions 
  within the framework of the time-dependent Ginzburg-Landau equations - the Gorkov-Eliashberg-Schmidt equations. We consider such solutions for lattices of arbitrary shape (in the extensive literature on the subject such solutions are considered only for triangular or rectangular) and for magnetic fields smaller than but close to the second critical magnetic field $H_{c2}=\kappa^2$ (the other case of magnetic fields larger than but close to the first critical magnetic field $H_{c1}$ treated in the literature is not addressed here) and, as common, for superconductors filling the entire $\R^2$.

We consider the simplest 
perturbations having the same (gauge-) periodicity as the underlying (stationary) Abrikosov vortex lattice solutions (we call such perturbations \emph{gauge-periodic}) and prove 
for a lattice of arbitrary shape that, under gauge-periodic perturbations,
\begin{itemize}
\item[(i)] \emph{ Abrikosov vortex lattice solutions are asymptotically stable for}  $\kappa^2 > \frac{1}{2}(1 - \frac{1}{\beta(\tau)})$;
 \item[(ii)]   \emph{Abrikosov vortex lattice solutions are unstable for} $\kappa^2 < \frac{1}{2}(1 - \frac{1}{\beta(\tau)})$.
\end{itemize}
\textbf{Here $\beta(\tau)$ is the Abrikosov constant depending on the lattice shape $\tau$ (see Subsection \ref{abr-lat} for the definition).} This result belies the common belief among physicists and mathematicians that  Abrikosov - type vortex lattice solutions are stable only for triangular lattices and $\kappa^2 > \frac{1}{2}$. Here $\tau$ is complex number parametrizing the lattice shapes. (For the definitions of various stability notions see Subsection \ref{subsec:stab}.)
Gauge - periodic perturbations are not common type of perturbations occurring in superconductivity, but the methods we develop are fairly robust  and can be extended - at the expense of significantly more technicalities - to substantially wider class of perturbation, which will be done elsewhere. Moreover, the same techniques could be used in other problems of pattern formation, which are ubiquitous in applications.

 To our knowledge,  previously, the only known result on stability
of the Abrikosov vortex lattice solutions concerned  orbital stability (under the same type of perturbations) which, though not stated explicitly,  can be deduced from the variational proof of \cite{Odeh} that the single vortex  Abrikosov lattices  for $\kappa^2>\frac{1}{2}$ are global minimizers of the Ginzburg-Landau energy functional on the fundamental cell (see also \cite{Dut}, both authors considered only triangular and rectangular lattices; the argument proving the orbital stability for global minimizers was proposed in  \cite{BL}).  However variational techniques do not give the asymptotic stability. Also variational techniques do not give even the orbital stability for $\kappa^2<\frac{1}{2}$ as in this case  Abrikosov lattices are not global minimizers of the Ginzburg-Landau energy functional on the fundamental cell.

In the rest of this section we present the basic equations involved, discuss their properties and related definitions and present our result.

\subsection{Gorkov-Eliashberg-Schmidt equations}

Macroscopically, states of the superconductors are described by the triples, $(\Psi, A, \Phi) : \R^+ \times \R^2 \to \C \times \R^2 \times \R$, where  $\Psi$ is the (complex-valued) order parameter, $A$ is the vector potential, and $\Phi$ is the scalar potential. Physically, $|\Psi|^2$ gives the (local) density of electrons having formed Cooper pairs. $\CURL A$ is the magnetic field, and $-\partial_t A - \nabla\Phi$ is the electric field.
 The dynamics is given by the Gorkov-Eliashberg-Schmidt equations (or time-dependent Ginzburg-Landau equations) in $\R^2$ which can be written as
\begin{equation}\label{GES}
\begin{cases}
    \gamma \partial_{t\Phi} \Psi = \COVLAP{A}\Psi + \kappa^2(1 - |\Psi|^2)\Psi, \\
    \CURL^*\CURL A = -\sigma \partial_{t\Phi} A + \Im(\bar{\Psi}\COVGRAD{A}\Psi).
\end{cases}
\end{equation}
Here $\kappa$ is a positive constant, $\gamma$ a complex number with $\Re \g>0$, $\sigma$ a positive $2\times 2$ matrix and $\COVGRAD{A} = \nabla - iA$ and $\COVLAP{A} = \COVGRAD{A}\cdot \COVGRAD{A}$ are the covariant gradient and Laplacian, and $\partial_{t\Phi}$ is the covariant time derivative $\partial_{t\Phi}\Psi = (\partial_t + i\Phi)\Psi$ or $\partial_{t\Phi}A = \partial_tA + i\nabla\Phi$, $\CURL A := \partial_{x_1}A_2 - \partial_{x_2}A_1$ and $\CURL^* f = (\partial_{x_2}f, -\partial_{x_1}f)$.
The second equation is Amp\`ere's law with $J_N = -\sigma (\partial_t A + \nabla\Phi)$ (using Ohm's law with $\sigma$ as the conductivity tensor) being the normal current associated to the electrons not having formed Cooper pairs, and $J_S = \Im(\bar{\Psi}\COVGRAD{A}\Psi)$ being the supercurrent associated to the electrons having formed such pairs. Note that for a solution $(\Psi, A, \Phi)$ of \eqref{GES} the pair  $(\Psi, A)$ determines $\Phi$ through the equation
\begin{equation}\label{phieq}
\Delta \Phi=-\partial_t \divv A+\divv\sigma^{-1}
[\Im(\overline{\Psi}\nabla_A \Psi)- \CURL^* \CURL A].
\end{equation}

The equations \eqref{GES} have the structure of a gradient-flow equation for
the Ginzburg-Landau energy functional given by
\begin{equation}\label{energy}
    \E_\Omega(\Psi, A) = \frac{1}{2} \int_\Omega \left\{ |\COVGRAD{A}\Psi|^2 + |\CURL A|^2 + \frac{\kappa^2}{2}(1 - |\Psi|^2)^2 \right\},
\end{equation}
where $\Omega$ is any domain in $\R^2$.
Indeed, (ignoring for the moment the boundary terms) they can be put in the form
\begin{equation}\label{grad-flow}
    \partial_{t\Phi}( \Psi,  A) =  - \lam\E'_\Omega(\Psi, A),
\end{equation}
where $\lam$ is the block-diagonal matrix given by $\lam:=\diag (\g^{-1}, \sigma^{-1})$,
and $\E'_\Omega$ is given by
\begin{equation}\label{grad}
    \E'_\Omega(\Psi, A) = (-\COVLAP{A}\Psi - \kappa^2(1 - |\Psi|^2)\Psi, \CURL^*\CURL A - \Im(\bar{\Psi}\COVGRAD{A}\Psi)),
\end{equation}
and is formally the $L^2$-gradient of $\E_\Omega(\Psi, A)$. Though \eqref{grad-flow} is not a standard form of the gradient flow, we show in Lemma \ref{lem:decreasing-energy} below that the energy \eqref{energy} decreases under the flow.



\bigskip


\subsection{Ginzburg-Landau equations}

The static solutions of the Gorkov-Eliashberg-Schmidt equations \eqref{GES} are triples $(\Psi, A, 0)$ independent of time. In this case $(\Psi, A)$ satisfies the well-known Ginzburg-Landau equations, which describe superconductors in thermodynamic equilibrium and which are given by
\begin{equation}\label{GL}
\begin{cases}
    -\COVLAP{A}\Psi = \kappa^2(1 - |\Psi|^2)\Psi, \\
    \CURL^*\CURL A = \Im(\bar{\Psi}\COVGRAD{A}\Psi).
\end{cases}
\end{equation}
Not surprisingly, they are the Euler-Lagrange equations for the energy functional \eqref{energy}: $\E'(\Psi, A) = 0$.

\subsection{Symmetries}

The Gorkov-Eliashberg-Schmidt equations~\eqref{GES} admit several continuous symmetries, that is, transformations which map solutions to solutions:

{\it Gauge symmetry}:  for any sufficiently regular function $\eta : \R^2 \to \R$,
\begin{equation*}
    (\Psi(x, t),\  A(x, t),\ \Phi(x, t))  \mapsto (e^{i\eta(x, t)}\Psi(x, t),\  A(x, t) + \nabla\eta(x, t),\   \Phi(x, t) - \p_t\eta(x, t));
\end{equation*}

{\it Translation symmetry}: for any $h \in \R^2$,
\begin{equation*}
    (\Psi(x, t),\  A(x, t),\ \Phi(x, t))  \mapsto (\Psi(x + h, t),\  A(x + h, t),\ \Phi(x + h, t));
\end{equation*}

{\it Rotation and reflection symmetry}: for any $R \in O(2)$ (including  the reflections $f(x) \mapsto f(-x)$)
\begin{equation*}
    (\Psi(x, t),\  A(x, t),\ \Phi(x, t)) \mapsto (\Psi(Rx, t),\   R^{-1}A(Rx, t),\ \Phi(Rx, t)).
\end{equation*}
These symmetries restrict to symmetries of the Ginzburg-Landau equations by considering time-independent transformations.


\subsection{Abrikosov lattices}\label{abr-lat}

The Abrikosov (vortex) lattice solutions are  solutions of the Ginzburg-Landau equations \eqref{GL} which exhibit double-periodicity in their physical properties. 
To give a rigorous definition, let $\LAT$ be a lattice. We say that $(\Psi, A)$ is gauge-periodic (with respect to the lattice $\LAT$), if there exist functions $g_\nu : \R^2 \to \R$, $\nu$ belonging to a basis of $\LAT$, such that for all basis vectors $\nu$,
\begin{equation}\label{ALper-cond}
\begin{cases}
    \Psi(x + \nu) = e^{ig_\nu(x)}\Psi(x), \\
    A(x + \nu) = A(x) + \nabla g_\nu(x).
\end{cases}
\end{equation}
(In terminology of \cite{ST} the pair $(\Psi, A)$ is \emph{equivariant} under the group of lattice translations.)
An important property of gauge-periodic pairs is the quantization of magnetic flux. Let $\Omega$ be any fundamental cell of $\LAT$. Then  the magnetic flux quantization property states that $ \int_\Omega \CURL A=2\pi n$ for some integer $n$. This can be written in terms of the average magnetic flux, $b = \lan \CURL A \ran_{\LAT}$ as
$    b = \frac{2\pi n}{|\Om|},$
where $|\Om|$ denotes the area of $\Omega$,  and, for $f$ any $\LAT$-periodic function, $\lan f \ran_\LAT = \frac{1}{|\Omega|} \int_\Omega f$, denotes the average per lattice cell. Using the reflection symmetry of the problem, one can easily check that we can always assume $n \geq 0$.

Now any lattice $\LAT$ can be given a basis $\{ \nu_1, \nu_2 \}$ such that the complex number $\tau = \frac{|\nu_2|}{|\nu_1|} e^{i\theta}$, where $\theta$ is the angle between $\nu_1$ and $\nu_2$, satisfies the conditions that $|\tau| \geq 1$, $\Im\tau > 0$, $-\frac{1}{2} < \Re\tau \leq \frac{1}{2}$, and $\Re\tau \geq 0$, if $|\tau| = 1$. Although the basis is not unique, the value of $\tau$ is, and we will call it the shape parameter of the lattice.

Note that $\tau$, $b$, and $n$ determine the lattice $\LAT$ up to rotation (which is a symmetry of the Ginzburg-Landau equations). We will say that a pair $(\Psi, A)$ is of type $(\tau, b, n)$, if the underlying lattice has shape parameter $\tau$, the average magnetic flux per lattice cell is equal to $b$, and there are $n$ quanta of magnetic flux per lattice cell. We also restrict ourselves to  $C^\infty$ pairs $(\Psi, A)$ which suffices for us due to elliptic and parabolic regularity. We have the following existence theorem (see \cite{Odeh, BGT, Dut, TS}).
\begin{theorem}
\label{thm:existence}
    Let $\tau$ be any lattice shape parameter and let $b$ be such that $\kappa^2 - b$ is sufficiently small. Then there exists an Abrikosov lattice solution $u^\tau_b = (\Psi^\tau_b, A^\tau_b)$ of type $(\tau, b, 1)$.
\end{theorem}
More detailed properties of these solutions are given in Section \ref{sec:rescaling} below. As we deal only with the case $n = 1$, we now assume that this is so and drop $n$ from the notation. Note that in this case the average magnetic field, $b$, and the fundamental cell area, $|\Om|$, are related as
\begin{equation}
    b = \frac{2\pi }{|\Om|}.
\end{equation}

Using the symmetries of the Ginzburg-Landau equations, one can show  (see \cite{TS} for example) that any Abrikosov lattice $(\Psi, A)$ of type $(\tau, b)$ is gauge equivalent to one satisfying the following conditions:
\begin{enumerate}[(I)]
\item $(\Psi, A)$ is gauge-periodic with respect to the lattice $\LAT^\tau_b$ spanned by $r^\tau_b(1,0)$ and $r^\tau_b(\Re\tau,\Im\tau)$, where $r^\tau_b = \sqrt{\frac{2\pi}{b\Im\tau}}$.
\item The function $g_\nu(x)$ in \eqref{ALper-cond} can be chosen as $g_\nu(x) = \frac{b}{2}\nu \cdot J x$, 
where  $J$ is the symplectic matrix
    \begin{equation*}
        J = \left( \begin{array}{cc} 0 & -1 \\ 1 & 0 \end{array} \right),
    \end{equation*}
     and therefore $\Psi$ and $A$ satisfy the quasiperiodic boundary conditions
    \begin{equation}
        \Psi(x+\nu) = e^{\frac{ib}{2} \nu\cdot J x}\Psi(x),\ \quad \textrm{and}\ \quad  A(x+\nu) = A(x)+\frac{b}{2}J\nu,
    \end{equation}
for any element  $\nu$ of a basis of $\LAT$.
   \item  $\divv A(x)=0$ and $\lan A(x) - A_b^0(x) \ran_{\LAT^\tau_b} = 0$, where $A_b^0(x):= \frac{b}{2}Jx$. 
\end{enumerate}
We note that the only continuous symmetry that preserves these properties is the $U(1)$ symmetry
\begin{equation*}
    T_\gamma : u= (\Psi, A) \ra T_\gamma u= (e^{i\gamma}\Psi, A),\ \forall \gamma \in \R.
\end{equation*}

\subsection{Gauge periodic perturbations}\label{subsec:pert}

We consider perturbations $v = (\xi, \alpha)$ of the $(\tau, b)$-Abrikosov lattice $u_b^\tau$ s.t.  $u_b^\tau+v$ is again of the type $(\tau, b)$. Hence it suffices to restrict the problem to any fundamental cell $\Omega^\tau_b$ of the lattice $\LAT^\tau_b$. From the definition of the type $(\tau, b)$ pairs, it follows that $\xi$ and $\al$ satisfy the quasiperiodic boundary conditions
\begin{equation}\label{xial-per-cond}
    \xi(x+\nu) = e^{\frac{ib}{2} \nu \cdot Jx}\xi(x),\ \quad  \textrm{and}\ \quad  \alpha(x+\nu) = \alpha(x),
\end{equation}
for any element  $\nu$ of a basis of $\LAT$, and with $\alpha$ 
being divergence-free, and having mean zero, i.e.,
\begin{equation}\label{dival-aval}
    \divv\alpha(x) = 0,\  \quad   \textrm{and}\  \quad \lan \alpha(x) \ran_{\LAT^\tau_b} = 0.
\end{equation}
 We introduce our space of perturbations $\Per^\tau_b$ to consist of all pairs $v = (\xi, \alpha) \in H^1(\Omega^\tau_b;\C\times\R^2)$ satisfying \eqref{xial-per-cond} and \eqref{dival-aval}. This space is naturally a (real) Hilbert space with the $H^1$ inner product of $v = (\xi, \alpha)$ and $v' = (\xi', \alpha')$ given by
\begin{align*}
    \langle v, v' \rangle_{H^1} = \frac{1}{|\Omega^\tau_b|} \Re \int_{\Omega^\tau_b} \big(\bar{\xi}\xi' + \overline{\COVGRAD{A_b^0}\xi} \cdot \COVGRAD{A_b^0}\xi' + \alpha \cdot \alpha' + \sum_{k=1}^{2} \nabla \alpha_k \cdot \nabla \alpha_k'\big).
\end{align*}
Here we use that the covariant gradient preserves the quasiperiodic boundary conditions and the dot product is in $\R^2$.  We also introduce the $L^2$ inner product
\begin{equation}
\label{eq:L2-inner-product}
    \langle v, v' \rangle_{L^2} = \frac{1}{|\Omega^\tau_b|} \Re \int_{\Omega^\tau_b} \big(\bar{\xi}\xi' + \alpha \cdot \alpha\big).
\end{equation}
Note that (i) if $u=(\Psi, A)$ satisfies (I) - (III), then $u-u_b^\tau$ satisfies \eqref{xial-per-cond} - \eqref{dival-aval};
and (ii) the conditions \eqref{xial-per-cond} and \eqref{dival-aval} break the translational and most of the gauge symmetry, leaving only the global gauge symmetry given by 
the gauge transformation 
\begin{equation*}
    T_\gamma : v= (\xi, \al) \mapsto T_\gamma v= (e^{i\gamma}\xi, \al).
\end{equation*}


\subsection{Stability under gauge periodic perturbations}\label{subsec:stab}

We now wish to study the stability of these Abrikosov lattice solutions under a class of perturbations that preserve the double-periodicity of the solution. More precisely, we focus on solutions $(\Psi, A, \Phi)$ of \eqref{GES} with the initial conditions of the form 
 $u_0\equiv(\Psi_0, A_0)=u^\tau_b + v_0,\ \ v_0 \in \Per^\tau_b$.

We now consider the tubular neighbourhood $U_\delta$ of the manifold $\{ T_\gamma u^\tau_b : \gamma \in \R \}$ of Abrikosov lattices, given by
\begin{equation}
\label{eq:tubular}
    U_\delta = \{ T_\gamma(u^\tau_b + v) : \gamma \in \R, \ v \in \Per^\tau_b, \ \|v\|_{H^1} < \delta \}
\end{equation}

We will say that $u^\tau_b$ is \textit{orbitally stable} under gauge-periodic perturbations if for all $\e > 0$, there exists $\delta > 0$ such that for any $C^1$ solution $(\Psi, A, \Phi)$, if $u = (\Psi, A) \in U_\delta$ when $t = 0$, then $u \in U_\e$ for all $t \geq 0$, and \textit{asymptotically stable}, if there exist $\gamma(t)\in \R$, such that any $C^1$ solution $(\Psi, A, \Phi)$, with $u = (\Psi, A) \in U_\delta$ when $t = 0$, satisfies $u - T_{\gamma(t)}u^\tau_b\ra 0,$ as $t \ra \infty$. We will say that $u^\tau_b$ is \textit{unstable} under gauge-periodic perturbations if it is not orbitally stable.

In order to state are main result, we first need to introduce the Abrikosov function $\beta(\tau)$. We fix $b = \kappa^2$ and define $\beta(\tau)$ to be
\begin{equation}
\label{eq:abrikosov-constant}
    \beta(\tau) =  \frac{\lan |\xi|^4 \ran_{\LAT^\tau_b}}{\lan |\xi|^2 \ran_{\LAT^\tau_b}^2},
\end{equation}
where 
$\xi \neq 0$ is the unique solution of the equation $-\COVLAP{A^0_b}\xi = \kappa^2 \xi$ satisfying the quasiperiodic boundary conditions \eqref{xial-per-cond}.
The main result in this paper is the following.
\begin{theorem}\label{thm:stability}
    For all $b$ sufficiently close to $\kappa^2$, the Abrikosov lattice $u^\tau_b$ is asymptotically stable under gauge-periodic perturbations if $\kappa^2 > \frac{1}{2}(1 - \frac{1}{\beta(\tau)})$ and is unstable if $\kappa^2 < \frac{1}{2}(1 - \frac{1}{\beta(\tau)})$.
\end{theorem}

To our knowledge there have been no prior results on the asymptotic stability of  Abrikosov lattices.
Orbital stability can often be deduced if the solution in question is a minimizer of an appropriate energy functional (see \cite{BL}). Hence the orbital stability of Abrikosov lattices for $\kappa^2>\frac{1}{2}$ follows from the variational proof of \cite{Odeh} that the single vortex  Abrikosov lattices  for $\kappa^2>\frac{1}{2}$ are global minimizers of the Ginzburg-Landau energy functional on the fundamental cell (see also \cite{Dut}, both authors considered only triangular and rectangular lattices).  However variational techniques do not give asymptotic stability. Also variational techniques do not give even orbital stability for $\kappa^2<\frac{1}{2}$ as in this case  Abrikosov lattices are not global minimizers of the Ginzburg-Landau energy functional on the fundamental cell.

This paper is organized as follows. In the next section we prove the theorem above, modulo a statement about properties of the Hessian of the Ginzburg-Landau energy functional. The latter statement is proven in Section \ref{sec:positivity}. In the Appendix we use energy methods to prove the weaker statement of orbital stability.


\section{Proof of Theorem \ref{thm:stability}}

As we consider $\tau$ fixed, from now on we do not display it in the notation. To simplify the notation we set $\g=1$ and $\s=\one$. To extend the proof of Theorem \ref{thm:stability} to   $\g$ and $\s$ satisfying $\Re\g>0$ and $\s>0$ is straightforward.

\subsection{Decomposition}

The one-dimensional manifold of solutions $\mathcal{M} = \{ T_\gamma u_b : \gamma \in \R \}$ lies in $U_\delta$ for all $\delta$. The tangent space at $u_b \in \mathcal{M}$ is spanned by the infinitesimal global gauge transformation, $\Gamma_b=\p_\g T_\g u_b|_{\g=0}$, given by
\begin{equation}\label{glob-gauge-mod}
    \Gamma_b = (i\Psi_b, 0).
\end{equation}
We now prove the following decomposition for $u$ close to the manifold $\mathcal{M}$.

\begin{proposition}\label{prop:decomposition}
    There exist $\delta_0 > 0$ and 
    a map $\eta : U_{\delta_0} \to \R$ such that $T_{\eta(u)}u - u_b \perp \Gamma_b$ (with respect to the $L^2$ inner product).
\end{proposition}
\begin{proof}
    Let $X$ be the affine space $X = u_b + \Per_b$. The desired $\g=\eta(u)$ solves the equation $f(\gamma, u)=0$, where the map
    $f : \R \times X \to \R$ is defined as
    \begin{equation*}
        f(\gamma, u) = \lan \Gamma_b,  T_{\gamma}u - u_b \ran_{L^2}.
    \end{equation*}
    We wish to apply the Implicit Function Theorem to $f(\gamma, u)=0$. It is clear that $f$ is a $C^1$ map and that $f(0,u_b) = 0$. A calculation shows that $D_\gamma f(0,u_b) = \int_{\Omega_b} |\Psi_b|^2 \neq 0$. The Implicit Function Theorem then gives us a neighbourhood $V$ of $u_b$ in $X$ and a neighbourhood $W$ of $0$ in $\R$ and a map $h : V \to W$ such that $f(\gamma, u) = 0$ for $(\gamma, u) \in W \times V$ if and only if $\gamma = h(u)$. We can always assume that $V$ is a ball of radius $\delta_0$.

    We can now define the map $\eta$ on $U_\delta$ for $\delta < \delta_0$ as follows. Given $u \in U_\delta$, choose $\gamma$ such that $u = T_\gamma u'$ with $ u' \in V$. We define $\eta(u) = h(u') - \gamma$.

    We show that $\eta$ is well defined. We first show that if $|\gamma|$ is sufficiently small, then $h(T_\gamma u) = h(u) - \gamma$. First note for all $\gamma$, $T_\gamma(V) \subset V$. One can easily verify that $f(h(u) - \gamma, T_\gamma u) = f(h(u), u) = 0$, and therefore by uniqueness of $H$, it suffices to show that $h(u) - \gamma \in W$, but this can easily be done by taking $V$ to be smaller if necessary.

    Suppose now that we also have $u = T_{\gamma'}(u_b + v')$. Then $T_{\gamma' - \gamma}(u_b + v') = u_b + v$. We therefore have
    \begin{equation*}
        h(u_b + v) - \gamma
        = h(T_{\gamma' - \gamma}(u_b + v')) - \gamma
        = h(u_b + v') - (\gamma' - \gamma) - \gamma
        = h(u_b + v') - \gamma',
    \end{equation*}
    so $\eta$ is well-defined. Finally, we compute $\lan \Gamma_b,  T_{\eta(u)}u - u_b \ran_{L^2}=\lan \Gamma_b,  T_{h(u')}T_{-\gamma}u - u_b \ran_{L^2}=\lan \Gamma_b,  T_{h(u')}u' - u_b \ran_{L^2}=f(h(u'), u')=0$ and the proof is complete.
\end{proof}


\subsection{Hessian}

The chief tool in the proof of the stability result is the analysis of the associated Hessian, i.e., the second derivative of the energy functional $\E$. We first note that $\E$ is a well-defined functional on $U_\delta$, and that $\E'$ is in fact its $L^2$-gradient in the sense that for all $v \in \Per_b$,
\begin{equation*}
    D\E(u)v = \lan \E'(u), v \ran_{L^2}.
\end{equation*}
where $D$ is the G\^ateaux derivative, $D\E(u)v = \p_s\E(u+sv)|_{s=0}$. We define the Hessian at $u_b$ to be the operator $L_b := D\E'(u_b)$, where $D$ is the G\^ateaux derivative on maps, on the domain $D(L_b): = \Per_b$.  Explicitly, $ \textrm{for}\ v = \left( \begin{array}{c}\xi\\  \alpha \end{array} \right)$, it is given by
\begin{equation} \label{L}
L_b v =
\left( \begin{array}{c} -\COVLAP{A_b}\xi + \kappa^2(2|\Psi_b|^2 - 1)\xi + \kappa^2{\Psi_b}^2\bar{\xi} + 2i(\COVGRAD{A^\tau_b}\Psi_b)\cdot\alpha \\
\CURL^*\CURL \alpha + |\Psi_b|^2\alpha - \Im(\bar{\xi}\COVGRAD{A_b}\Psi_b + \bar{\Psi}_b\COVGRAD{A_b}\xi) \end{array} \right).
\end{equation}
This is a real-linear operator on the space $L^2(\Om, \C)\times L^2(\Om, \R^2)$ with the domain $\cP_b$. We define for it the notion of spectrum and of discrete spectrum in the usual way.
 We introduce the new parameter 
\begin{equation}\label{parameter}
	\e = \sqrt{\frac{\kappa^2 - b}{\kappa^2[(2\kappa^2 - 1)\beta(\tau) + 1]}}.
\end{equation}
The term $(2\kappa^2 - 1)\beta(\tau) + 1$ in denominator is necessary in order to have a positive expression under the square root and to regulate the size of the perturbation domain. 
 \DETAILS{$\e=\sqrt{\frac{\kappa^2 - b}{\kappa^2\lam^1}} $, where
    \begin{align} \label{lam1}
\lam^1 = \left(\kappa^2 - \frac{1}{2} \right) \frac{\lan |\psi^0|^4\ran_\cL}{\lan |\psi^0|^2\ran_\cL } + 
\frac{1}{2}\lan |\psi^0|^2\ran_\cL,
	\end{align}
with  $\psi^0$ satisfying the equation  $-\COVLAP{A^0_b}\psi^0 = \kappa^2 \psi^0$ and normalized as ??.
(The normalization constant $\lam^1$ is introduced for convenience.)}
 The main result concerning this operator which we use in the proof  of Theorem \ref{thm:stability} is the following theorem, whose proof we postpone until Section \ref{sec:positivity}.
\begin{theorem} \label{thm:positivity}
    Suppose that $\kappa^2 \neq \frac{1}{2}(1 - \frac{1}{\beta(\tau)})$ and that $b$ is sufficiently close to $\kappa^2$. Then we have the following statements:
    \begin{enumerate}
    \item The operator $L_b$ has a real, discrete spectrum which includes the eigenvalue $0$ with the eigenfunction $\Gamma_b$, while its lowest eigenvalue, $\theta$, on the subspace  $\{v \in \Per_b\ |\  v \perp \Gamma_b\}$ is of the form
\begin{equation} \label{theta}
\theta:= 2b \left[ \left(2\kappa^2 - 1\right)\beta(\tau) + 1 \right] \e^2 +O(\e^3).
\end{equation}
Consequently,  if $\kappa^2 < \frac{1}{2}(1 - \frac{1}{\beta(\tau)})$, then $L_b$ has a negative eigenvalue.  

\item If $\kappa^2 > \frac{1}{2}(1 - \frac{1}{\beta(\tau)})$, then there is a uniform constant $c > 0$, such that for all $v \in \Per_b$ satisfying $v \perp \Gamma_b$,
        \begin{equation} \label{L-low-bound}
           \lan v, L_b v \ran_{L^2} \geq c\theta \|v\|_{H^1}^2.
        \end{equation}

 \item  There exists a positive constant $c> 0$ such that, for all $v \in \Per_b$,
\begin{equation}  \label{L-upp-bound}
 |\lan v, L_b v \ran_{L^2}| \le c \|v\|_{H^1}^2. 
\end{equation}

    \end{enumerate}
\end{theorem}


\bigskip

\subsection{Asymptotic stability}

We now assume that we have the case $\kappa^2 > \frac{1}{2}(1 - \frac{1}{\beta(\tau)})$ and prove the asymptotic stability of the Abrikosov lattice $u_b$. To this end we derive and use differential inequalities for  the Lyapunov functional
\be \la{Lyap}
\Lam(t)
  =\frac{1}{2}\lan v, L_b v\ran_{L^2}. 
\end{equation}
\begin{lemma}\label{lem:v-control'}
Let  $(\Psi(t), A(t), \Phi(t))$ be a solution of the Gorkov-Eliashberg-Schmidt equations \eqref{GES} on the time interval $[0,T]$ such that $u(t) := (\Psi(t), A(t))$ satisfies $T_{\gamma(t)}u(t) = u^\tau_b + v(t)$, $\gamma(t) \in \R$, $v \in C^1([0,T]; \Per_b)$, $v(t) \perp \Gamma_b$ and $\|v(0)\|_{H^1} \ll 1$. Then  $\|v(t)\|_{H^1}\le e^{-\delta t/2}\|v(0)\|_{H^1}$ for $\delta < \theta$ for all $t \in [0,T]$, where $\theta$ is as in Theorem \ref{thm:positivity}.
\end{lemma}
\begin{proof}
 We plug the decomposition $u(t) = T_{-\gamma(t)}(u_b + v(t))$ into \eqref{GES} 
and use the covariance of this equation with respect to the transformation $ T_{\gamma(t)}$ to obtain the  equation
\be \la{lineareq}
\partial_{t}v 
  =-L_b v +N(v), 
\end{equation}
where $N(v)$ is the nonlinearity given by
 \begin{equation}\label{N}
N(v)=\left( \begin{array}{c} (\n_{A^b_0}\cdot\al+\al\cdot\n_{A^b_0}+|\al|^2)\xi + |\al|^2\Psi_b -\kappa^2(2\Re(\bar\Psi_b\xi)-|\xi|^2)\xi +  \kappa^2|\xi|^2{\Psi_b} \\
  \Im(\bar{\xi}\COVGRAD{A_b}\Psi_b + \Im(\bar{\xi}\COVGRAD{A^b_0}\xi-i\al(\bar{\xi}\Psi_b+\bar \Psi_b\xi) -i\al |\xi|^2) \end{array} \right)
\end{equation}
$ \textrm{for}\ v = \left( \begin{array}{c}\xi\\  \alpha \end{array} \right).$
Using equation \eqref{lineareq}, we obtain
\be \la{endescr}
-\p_t\Lam(t)   =\lan v, L_b^2 v\ran_{L^2} -\lan N(v), L_b v\ran_{L^2} .
\end{equation}
Using the expression for $N(v)$ and Sobolev embedding theorems we obtain easily
 the following rough estimate:
\begin{equation}\label{Nest''}
\|N(v)\|_{L^2} \lesssim \|v\|_{H^2}(\|v\|_{H^1}+\|v\|_{H^1}^2),
\end{equation}
which 
implies that
 \begin{equation}\label{Nest}
|\lan N(v), L_b v\ran_{L^2}| \lesssim \| L_b v\|_{L^2}\|v\|_{H^2}(\|v\|_{H_1}+\|v\|_{H^1}^2).
\end{equation}
Since, by \eqref{L-upp-bound}, $\|v\|_{H^2} \lesssim \| L_b v\|_{L^2}+\|v\|_{L^2}$ and $\|v\|_{H^1}^2 \lesssim \Lam(t)$, the last inequality implies
 \begin{equation}\label{Nest'}
|\lan N(v), L_b v\ran_{L^2}| \lesssim (\| L_b v\|^2_{L^2}+\Lam(t))(\Lam(t)^{\frac{1}{2}}+\Lam(t)).
\end{equation}
Next, we think of $L_b$ as the restriction to the orthogonal complement of $\Gamma_b$ and define $L_b^{\al}$, for $0 < \alpha < 1$, by the formula
$L_b^{\al}=C\int_0^\infty (\frac{1}{\w} -\frac{1}{L_b+\w})\w^\al d \w$, where  $C^{-1}:=\int_0^\infty (\frac{1}{\w} -\frac{1}{1+\w})\w^\al d \w$. (Another way to proceed is to use the complex-linear extension of $L_b$ constructed in Subsection \ref{subsec:compl}.) Then writing $\lan v, L_b^2 v\ran_{L^2}=\lan L_b^{\frac{1}{2}}v, L_b L_b^{\frac{1}{2}}v\ran_{L^2}$ and using \eqref{L-low-bound}, we find $\lan v, L_b^2 v\ran_{L^2}\ge \theta\lan v, L_b v\ran_{L^2}=2\theta \Lam(t)$. Using this, we obtain
\be \la{Lyap-DI}
-\p_t \Lam(t)    \ge \theta\Lam(t)+\big[\frac{1}{2}-c(\Lam(t)^{\frac{1}{2}}+\Lam(t))\big]\| L_b v\|^2_{L^2}- \Lam(t)^{\frac{3}{2}}-\Lam^2(t).
\end{equation}
If we now assume that $\Lam(t)\ll 1$, then this gives
\be \la{Lyap-DI'}
-\p_t(e^{\delta t}\Lam(t))    \ge(\theta-\delta)e^{\delta t}\Lam(t).
\end{equation}
Integrating the last inequality from $0$ to $t$, one finds
 \be \la{Lyap-est}
 \Lam(0)  \ge e^{\delta t}\Lam(t) +(\theta-\delta)\int_0^te^{\delta s}\Lam(s)d s
\end{equation}
and in particular $\Lam(t)\le e^{-\delta t}\Lam(0)$ for $\delta < \theta$. Taking  $\Lam(0))\ll 1$, we see that our assumption $\Lam(t))\ll 1$ is justified, which completes the argument. Finally appealing again to \eqref{L-upp-bound} shows that $\|v(t)\|_{H^1}\le e^{-\delta t/2}\|v(0)\|_{H^1}$ for $\delta < \theta$.
\end{proof}
To finish the proof of asymptotic stability, let $\delta_0$ be given by Proposition \ref{prop:decomposition} and let $u(0) \in U_{\frac{1}{2}\delta_0}$. By standard parabolic existence theory (see e.g. \cite{LSU, Lieberman}) for the equation \eqref{lineareq}, written for the real and imaginary parts of $v$, 
has a unique solution $u(t) \in U_{\delta_0}$ for $t\le T$, for some $T > 0$. Let $T_*$ be the supremum of such $T$. If $T_*< \infty$, then 
$u(T_*) \in \p U_{\delta_0}$. Then by Proposition \ref{prop:decomposition}, there is $\gamma(t) \in \R$ so that $T_{\gamma(t)}u(t) = u_b + v(t)$, with $v(t) \in C^1([0,T]; \Per_b)$ and $v(t) \perp \Gamma_b$. By Lemma \ref{lem:v-control'}, $\|v(t)\|_{H^1}\le e^{-\delta t/2}\|v(0)\|_{H^1}\le \frac{1}{2}\delta_0$ for all $t \in [0,T_*]$, which contradicts $u(T_*) \in \p U_{\delta_0}$. Hence $T_*= \infty$ and $u(t) \in U_{\delta_0/2}$ for all $t$ and $\|v(t)\|_{H^1}\le e^{-\delta t/2}\|v(0)\|_{H^1} \ra 0$, as $t\ra\infty$. Since $u(t) = T_{-\gamma(t)}u_b + T_{-\gamma(t)}v(t)$, we obtain (after replacing $-\gamma(t)$ by $\gamma(t)$)
 \be \la{as-stab}
 \|u(t) - T_{\gamma(t)}u_b\|_{H^1}\le e^{-\delta t/2}\|u(0) - T_{\gamma(t)}u_b\|_{H^1} \ra 0,
\end{equation}
 as $t\ra\infty$, and this completes the proof of asymptotic stability.
 $\Box$
\DETAILS{We can also use a bootstrap starting with the Lyapunov functional
\be \la{enexp'}
\Lam'(t)
  =\frac{1}{2}\lan v,  v\ran_{L^2} , 
\end{equation}
and then using a bootstrap, possibly with an estimate on $\|v\|_{H_1}$ obtained by the energy method.}


\subsection{Proof of instability}

We now assume that we have the case $\kappa^2 < \frac{1}{2}(1 - \frac{1}{\beta(\tau)})$ and prove the instability of the Abrikosov lattice $u_b$. Let $\eta \in \Per_b$ be the negative eigenvector of $L_b$, corresponding to the eigenvalue $-\lam = \theta < 0$, given in Theorem \ref{thm:positivity}, so that $L_b \eta = - \lam \eta$. We normalize $\eta$ so that $\|\eta\|_{L^2} = 1$.

For $\delta > 0$ we now define $u(t)$ to be the solution with the initial datum $u_{\delta,0} = u_b + \delta \eta$. We write this solution as $u(t) = u_b + v_\delta(t)$. Then $v_\delta(t)$ satisfies the equation \eqref{lineareq} with the initial condition $v_{\delta,0} =  \delta \eta$. Using the Duhamel principle and the fact that $e^{-L_b t} \eta= e^{-\lam t} \eta$ we rewrite the latter equation in the form
\begin{equation} \label{duh}
v_\delta(t) =  \delta e^{\lam t}\eta + \int_{0}^{t}
e^{-L_b (t - s)}N(v_\delta(s))\ ds,
\end{equation}
where 
$N(v)$ is the nonlinearity given in \eqref{N}.
It satisfies the following estimate:
\begin{eqnarray} \label{non-est}
  \| N(v)\|_{ L^2} \lesssim \|v\|_{H^1}^2 + \|v\|_{H^1}^3 .
\end{eqnarray}
Next, 
for an appropriate large constant $c>\lam$,
we have by the standard elliptic theory, similarly to \eqref{L-upp-bound},
$\|  w \|_{H^1}^2 \le  \lan w, (L_b +c)w \ran_{ L^2}$.
The self-adjoint complex-linear extension of the operator $L_b$ obtained in the next section and the invariance of the image of $\cP_b$ under this extension imply the spectral decomposition for $L_b$, which gives that $\lan w, (L_b +c)e^{-2L_b t}w\ran \le (-\lam +c)e^{-2\lam t}\| w \|_{ L^2}^2$. The last two estimates imply the bound $\| e^{-L_b t} w \|_{H^1} \lesssim  e^{\lam t}\| w \|_{ L^2}$. Using the latter bound and \eqref{non-est} and writing in the rest of the proof $\| v\|$ for $\| v \|_{H^1}$, we obtain
\begin{eqnarray} \label{est1}
\| v_\delta(t)- \delta e^{\lam t} \eta \| &\le &
\int_0^t \| e^{-L_b (t-s)} N(v_\delta(s)) \| ds \\
&\lesssim &  \int_0^t e^{\lam (t-s)} \| {N}(v_\delta(s)) \|_{ L^2} ds \\
& \lesssim &   \int_0^t e^{\lam (t-s)}
(\|v_\delta(s)\|^2 + \|v_\delta(s)\|^3) ds.
\end{eqnarray}

Now, let $M$ be a constant satisfying 
$0<M<\| {\eta} \| =1$ 
and define
\begin{equation} \label{T1def}
T_1:= \sup \{ s: 
\| v_\delta(s) - \delta e^{\lam s} \eta \| \le M \delta e^{\lam s} \}.
\end{equation}
Clearly, $T_1>0$, and so for $0 \le t \le T_1$, we have by
triangle inequality
\begin{equation} \label{eq:bound}
\| v_\delta(s) \| \le \| v_\delta(s)-  \delta e^{\lam s} {\eta}  \| +
\delta e^{\lam s} \| \eta \| \le (M+1) \delta e^{\lam s}.
\end{equation}
Therefore, by  \eqref{est1} and \eqref{eq:bound}, we have for $0 \le t \le T_1$,
\begin{eqnarray}
\| v_\delta(t)-  \delta e^{\lam t} \eta  \| &\lesssim &
 \int_0^t e^{\lam (t-s)} (\delta^2 e^{2\lam s} + \delta^3 e^{3\lam s})ds   \nonumber \\
& = &   \delta^2 e^{\lam t} \int_0^t e^{\lam s} (1+ \delta e^{\lam s}) ds.  \nonumber
\end{eqnarray}
We  choose $T_2$ to satisfy $\delta e^{\lam T_2} =1$. Then
\begin{equation} \label{eq:key_estimate3}
\| v_\delta(t)-  \delta e^{\lam t} \eta  \|\le    C \delta^2 e^{2\lam t},
\;\;\;\;\mbox{for} \; \; 0 \le t \le \min(T_1, T_2).
\end{equation}
Pick $C' \ge M,\ C$ and define the constant $T^\delta>0$ by the relation
\begin{equation} \label{eq:onehalf2}
  C' \delta e^{\lam T^\delta}=M.
\end{equation}
Note that, since $C' \ge M$, we have $T^\delta < T_2$. 
We claim that $T^\delta \le T_1$.  If not, then $T_1 < T^\delta \le T_2$, and by
\eqref{eq:key_estimate3} and \eqref{eq:onehalf2},
\begin{eqnarray*}
\| v_\delta(T_1) -  \delta e^{\lam T_1} \eta  \| <
M \delta e^{ \lam T_1}.
\end{eqnarray*}
But this result contradicts the definition of $T_1$ in \eqref{T1def}.
Hence $T^\delta \le T_1$, and therefore $T^\delta \le \mbox{min}(T_1, T_2)$. Now we have by the triangle inequality, \eqref{eq:key_estimate3}, \eqref{eq:onehalf2} and the condition $\|\eta \|=1$,
\begin{eqnarray*}
\| v_\delta(T^\delta) \| & \ge & \| \delta e^{\lam T^\delta}
\eta \| - \| v_\delta(T^\delta) -  \delta e^{\lam
T^\delta} \eta  \| \\
& \ge & \delta e^{ \lam T^\delta}  - \delta
e^{ \lam
T^\delta}  M.
\end{eqnarray*}
If we set $\nu := (1 - M) \frac{M}{C'}>0$, then the last equation, together with  \eqref{eq:onehalf2}, implies $\| v_\delta(T^\delta) \|  \ge \nu$, which can be rewritten as 
\begin{equation}\label{est2}
    \| u(T^\delta) - u_b \| \geq \nu.
\end{equation}
Now we note that for $\delta$ sufficiently small, 
the unique minimizer $\g_*$ of $\| u - T_\gamma u_b \|^2$ satisfies  $\g_*=O(\delta)$ and therefore \eqref{est2} implies
\begin{equation}\label{est3}
    \inf_{\g}\| u(T^\delta) - T_\gamma u_b \| \geq \nu-O(\delta).
    \end{equation}
In other words, $\forall \delta>0$ sufficiently small, there is $u_0\in U_\delta$ such that $u(T^\delta)\notin U_{\frac{1}{2}\nu},\ \forall t\ge 0,$ for a fixed $\nu>0$ independent of $\delta$.
\DETAILS{Since $\eta \perp \Gamma^\tau_b$, we have that for all $\gamma$,
\begin{align*}
    \| u(t) - T_\gamma u_b \|_{L^2}
    &\geq \| u_b + \delta e^{\lambda t}\eta - T_\gamma u_b \|_{L^2}
    - \| u(t) - (u^\tau_b + \delta e^{\lambda t}\eta) \|_{L^2} \\
    &\geq \delta e^{\lambda t} - \frac{1}{2} \delta e^{\lambda t}
    = \frac{1}{2} \delta e^{\lambda t},
\end{align*}
as long as $\delta e^{\lambda t} \|\eta\|_{H^1} < \delta_0$. This implies instability.}
 This implies instability.
\bigskip


\section{Estimates on Hessian}
\label{sec:positivity}

In this section we prove Theorem \ref{thm:positivity} concerning the positivity of the Hessian $L_b$. In what follows we omit the subindex $L^2$ in the inner products and norms. 


\subsection{Shifted Hessian}

We first consider a shifted Hessian $\tilde{L}_b$, which induces the same quadratic form. Note that since $\DIV \alpha = 0$, we have $\CURL^* \CURL \alpha = -\Delta \alpha$. Again using $\DIV \alpha = 0$, we can integrate by parts to see that
\begin{equation*}
    \int_{\Omega_b} - \alpha\cdot\Im(\bar{\xi}\COVGRAD{A_b}\Psi_b + \bar{\Psi}_b\COVGRAD{A^\tau_b}\xi) )
    = \int_{\Omega_b} - 2 \alpha \cdot \Im(\bar{\xi}\COVGRAD{A_b}\Psi_b).
\end{equation*}
Therefore we have $\lan v, L_b v \ran_{L^2} = \lan v, \tilde{L}_b v \ran_{L^2}$ where $\tilde{L}_b$ is the operator on the same domain $\cP_b$ given by
\begin{equation}\label{hessian-shifted}
    \tilde{L}_b v =
    \left( \begin{array}{c}  -\COVLAP{A_b}\xi + \kappa^2 (2|\Psi_b|^2 -1)\xi + \kappa^2{\Psi_b}^2\bar{\xi} + 2i(\COVGRAD{A_b}\Psi_b)\cdot\alpha, \\
    -\Delta \alpha + |\Psi_b|^2\alpha - 2 \Im(\bar{\xi}\COVGRAD{A_b}\Psi_b) . \end{array}\right)
\end{equation}
It will be more convenient to work with $\tilde{L}_b$.

\subsection{Rescaling}
\label{sec:rescaling}

We now rescale the problem in order to exploit the analytic properties of the Abrikosov lattice solutions. Given a pair $(\Psi, A)$ of type $(\tau, b)$, we set $\sigma := \sqrt{\frac{1}{b}} = r^\tau_b \sqrt{\frac{\Im\tau}{2\pi}}$, and introduce the rescaling $U_\s : (\Psi(x), A(x)) \mapsto (\s \Psi(\sigma x), \s A(\sigma x))$. This has the effect that the rescaled state, $(\psi, a) := U_\s (\Psi, A)$, is of type $(\tau, 1)$.
It is easy to verify that $U_\s$ is a linear unitary bijection between $\Per_b$ and $\Per_1$ (in particular, it preserves the $L^2$ inner-product, i.e., $ \lan U_\s v, U_\s v' \ran = \lan v, v' \ran.$
Moreover it preserves the orthogonality condition in the sense that for $v \in \Per_b$, $v \perp \Gamma_b$ if and only if $U_\s v \perp \Gamma_1$).

We note that the rescaled Abrikosov lattice solution $(\psi_b, a_b) := U_\s (\Psi_b, A_b)$ satisfies the rescaled Ginzburg-Landau equations
\begin{equation}\label{glResc}
\begin{cases}
    (-\COVLAP{a} - \lam_b)\psi + \kappa^2 |\psi|^2\psi=0, \\
    \CURL^*\CURL a - \Im(\overline{\psi}\COVGRAD{a} \psi)=0,
\end{cases}
\end{equation}
where $\lam_b = \frac{\kappa^2}{b}$, and the quasiperiodic boundary conditions
\begin{equation}\label{ALpercond}
    \psi_b(x + \nu) = e^{\frac{i}{2} \nu \cdot J x} \psi_b(x),\ \quad  \textrm{and}\ \quad  a_b(x+\nu) = a_b(x) +\frac{1}{2}J\nu,
\end{equation}
where $\nu$ is either of the basis vectors of $\LAT_1$, as well as $\divv a_b=0$ and $\lan a_b - \frac{1}{2}Jx \ran_{\cL_1}=0$. 
Thus $(\psi_b, a_b)$ are of type $(\tau, 1)$. 

We now define the rescaled Hessian to be $L_b^{\textrm{resc}} : =\s^2 U_\s \tilde{L}_b {U_\s}^{-1}$. With $v = \left( \begin{array}{c} \xi\\ \alpha \end{array}\right)$, it is explicitly given by
\begin{equation}\label{eq:rescaled-hessian}
    L_b^{\textrm{resc}} v = \left( \begin{array}{c}   -\COVLAP{a_b}\xi - \lambda_b \xi + 2\kappa^2|\psi_b|^2\xi + \kappa^2{\psi_b}^2\bar{\xi} + 2i\alpha\cdot\COVGRAD{a_b}\psi_b\\
     -\Delta \alpha + |\psi_b|^2 \alpha - 2\Im(\bar{\xi}\COVGRAD{a_b}\psi_b)  \end{array}\right).
\end{equation}

For the rest of this section we 
 write $\cL,\ \Om,\ \Per$ for $\cL_1,\ \Om_1,\ \Per_1$.


\subsection{Complexification}
\label{subsec:compl}

 In order to freely use the spectral theory, it is convenient to pass from the real-linear operator $L_b^{\textrm{resc}}$ to a complex-linear one. To this end we complexify the space $\Per$ and extend the operator $L_b^{\textrm{resc}}$ to the new spaces. We first identify $\alpha : \R^2 \to \R^2$ with the function $\alpha^\C = \alpha_1 - i\alpha_2 : \R^2 \to \C$.
(Whenever it does not cause confusion we drop the $^\C$ superscript from the notation.)
We note that $\alpha \cdot \alpha' = \Re(\bar{\alpha}^\C {\alpha'}^\C)$. We also introduce the differential operator $\partial = \partial_{x_1} - i\partial_{x_2}$. We note that $\bar{\partial} \alpha^\C = \DIV\alpha - i\CURL\alpha$. where the $\bar{\partial}$ denotes complex conjugate operator. In general, for an operator $A$, we write $\bar{A} := \mathcal{C} A \mathcal{C}$, where $\mathcal{C}$ denotes complex conjugation.

We now consider the complex Hilbert space $L^2(\Om, \C^4)$ of vectors $(\xi, \phi, \alpha, \omega)$, with the usual $L^2$ inner product
\begin{equation*}
    \lan v, v' \ran = \lan \bar{\xi}\xi' + \bar{\phi}\phi' + \bar{\alpha}\alpha' + \bar{\omega}\omega' \ran_\LAT
\end{equation*}
The original space $L^2(\Om, \C)\times L^2(\Om, \R^2)$, on which $L_b^{\textrm{resc}}$ is defined, is embedded in $L^2(\Om, \C^4)$ via the injections
\begin{equation}\label{embed}
\pi_\pm:\	(\xi, \alpha) \mapsto 
\frac{1}{\sqrt{2}} (\xi, \pm\bar{\xi}, \alpha, \pm\bar{\alpha}),
\end{equation}
with inverses, $\pi_\pm^{-1}$, given by the obvious projection. $V_\pm:=\Ran \pi_\pm$ are real spaces spanning  $L^2(\R^2, \C^4)$:
\begin{equation}\label{Vpm-deco}
    (\xi, \phi, \alpha, \omega) = \frac{1}{2}(\xi + \bar{\phi}, \bar{\xi} + \phi, \alpha + \bar{\omega}, \bar{\alpha} + \omega) + \frac{1}{2}(\xi - \bar{\phi}, -\bar{\xi} + \phi, \alpha - \bar{\omega}, -\bar{\alpha} + \omega).
\end{equation}
Moreover, the map $I: (\xi, \phi, \alpha, \omega)\ra  (i\xi, i\phi, i\alpha, i\omega)$ acts between $V_\pm$: $I: V_\pm \ra V_\mp$. This embedding transfers the operator $L_b^{\textrm{resc}}$ to the range, $V_\pm:=\Ran \pi_\pm$, of this  injection. We denote the resulting operator by $L_b^{\textrm{transf}}$. Its domain is $\pi\Per$. We want to extend it to $L^2(\R^2, \C^4)$. 
To this end it is convenient to rewrite the operator $L_b^{\textrm{transf}}$ 
in complex notation. We introduce the notation $\partial_{a^\C} = \p - ia^\C$. Straightforward calculations show that
\begin{equation*}
	2i\alpha\cdot\COVGRAD{a_b}\psi_b
		= -i(\partial_{a_b^\C}^*\psi_b)\alpha^\C + i(\partial_{a_b^\C}\psi_b)\bar{\alpha}^\C,
\end{equation*}
and that
\begin{equation*}
	- \Im(\bar{\xi}\COVGRAD{a_b}\psi_b)^\C
		= \frac{i}{2}(\overline{\partial_{a_b^\C}^*\psi_b})\xi
				+ \frac{i}{2}(\partial_{a_b^\C}\psi_b)\bar{\xi}.
\end{equation*}
Using the above relations we rewrite the operator $L_b^{\textrm{transf}}$ and then define its complex-linear extension, denoted by $K_b$, by the resulting matrix
\begin{equation}\label{K}
	K_b = \left( \begin{array}{cccc}
		-\COVLAP{a_b} - \lambda_b + 2\kappa^2|\psi_b|^2
					& \kappa^2\psi_b^2
					& -i(\partial_{a_b}^*\psi_b) & i(\partial_{a_b}\psi_b) \\
		\kappa^2\bar{\psi}_b^2
				& -\overline{\COVLAP{a_b}} - \lambda_b + 2\kappa^2|\psi_b|^2
					& -i(\overline{\partial_{a_b}\psi_b}) & i(\overline{\partial_{a_b}^*\psi_b}) \\
		i(\overline{\partial_{a_b}^*\psi_b})
			& i(\partial_{a_b}\psi_b)
			& -\Delta + |\psi_b|^2 & 0 \\
		-i(\overline{\partial_{a_b}\psi_b})
			& -i(\partial_{a_b}^*\psi_b)
			& 0 & -\Delta + |\psi_b|^2
		\end{array} \right)
\end{equation}
on the domain which 
consists of all $v = (\xi, \phi, \alpha, \omega)\in H^2(\R^2, \C^4)$, with $\xi$, $\bar\phi$, and $\alpha$, $\bar\omega$ satisfying the quasi-periodic boundary conditions
\begin{equation}\label{ALpercond'}
    \chi(x + \nu) = e^{\frac{i}{2} \nu \cdot J x} \chi(x),\ \quad  \textrm{and}\ \quad  \s(x+\nu) = \s(x),
\end{equation}
where, as above, $\nu$ is either of the basis vectors of $\LAT$, as well as $\alpha$ and $\omega$ are divergence free and have mean-zero,  $\divv \s=0$ and $\lan \s\ran_{\cL}=0$.  
(Note that similarly to the Riesz - Fischer $L^2-$space on a torus (see e.g. \cite{McO}), we could have used results of \cite{TS} to introduce $L^2-$space on $\Om$ satisfying the quasiperiodic conditions \eqref{ALpercond'}, rather than periodic ones.)

The operator $K_b$ is clearly complex-linear, self-adjoint, has purely discrete spectrum, and, as it is not hard to check, satisfies
\begin{equation}\label{Kb-restr}
V_\pm\ \textrm{are invariant under}\ K_b\ \quad  \textrm{and}\ \quad    K_b|_{V_+}=L_b^{\textrm{transf}},
\end{equation}
\begin{equation}\label{Kb-deco}
K_b= K_b|_{V_+}+K_b|_{V_-},\ \quad \s(K_b|_{V_+})=\s(K_b|_{V_-})\ \quad \textrm{and}\  \ \quad \s(K_b)=\s(K_b|_{V_+})\cup \s(K_b|_{V_-}).
\end{equation}
For the second equation, we used \eqref{Vpm-deco} and that $K_b$ obviously commutes with $I$. \eqref{Kb-restr} implies
\begin{equation}\label{qfrel}
\lan v, L_b^{\textrm{resc}} v \ran =	\lan \pi_+ v, L_b^{\textrm{transf}} \pi_+ v \ran = \lan \pi_+ v, K_b \pi_+ v \ran.
\end{equation}
%

%
%
\DETAILS{\subsection{Creation and annihilation operators}

The operators $\al_0^*:=-\partial_{a_0}$ and $\al_0:=-\partial_{a_0}^*$ are the creation and annihilation operators,  respectively, associated to $\COVLAP{a_0}$:
\begin{equation*}
	-\COVLAP{a_0}+i\p^*a_0= \al_0^*\al_0,\ [\al_0, \al_0^*]=2\CURL a_0.
\end{equation*}
(Recall that, due to rescaling, $\CURL a_0=n=1$ and note that to go to the complex representation we use
$\CURL a +i\div a= i\bar\p a^\C$. 
In general, we have that
 \begin{equation*}
	-\COVLAP{a}-\CURL a= \partial_{a}\partial_{a}^*,\ [\partial_{a}^*, \partial_{a}]=2\CURL a.
\end{equation*}}
%

\subsection{Perturbation theory}

 It is shown in \cite{TS} that for each $\tau$ there is $\e_0> 0$, such that the solutions $(\psi_b, a_b)$  form a real-analytic branch of solutions in the (bifurcation) parameter $\e \in [0,\e_0] $ defined in \eqref{parameter} and
have the following expansions for 
(see \cite{TS} where a more general result was derived):
\begin{equation*}
	\begin{cases}
		\psi_b = \e \psi^0 + \e^3 \psi^1 + O(\e^5), \\
		a_b = a^0 + \e^2 a^1 + O(\e^4),
	\end{cases}
\end{equation*}
where 
$\psi^0$ and $a^1$ satisfy the following relations
\begin{equation}\label{c0vac}
	\partial_{a^0}^*\psi^0 =0,\  \langle |\psi^{0} |^2 \rangle_\cL =2, 
\end{equation}
\begin{equation}\label{curla1}
	i\bar\p a^1
		= \frac{1}{2} (\lan |\psi^0|^2 \ran_\cL - |\psi^0|^2),
\end{equation}
\begin{equation}\label{a1eq}
    \Delta a^1 = \frac{i}{2}\bar{\psi}^0 (\partial_{a^0}\psi^0).
\end{equation}
 We now prove the following proposition.
\begin{proposition}\label{prop:qfKbnd}
    Let $\Gamma_b^c$ be the gauge zero mode in the extended space: $\Gamma_b^c = (i\psi_b, -i\bar{\psi}_b, 0, 0)\in V_+$. For $\e>0$ sufficiently small, the lowest eigenvalue, $\theta$, of $K_+:=K_b|_{V_+} \equiv L_b^{\textrm{transf}}$ on the subspace $\{v \in D(K_+)\ |\ v \perp \Gamma\}$ is of the form
    \begin{equation}\label{lowest-ev-K}
        \theta:= 2\left[\left(2\kappa^2 - 1\right)\beta(\tau) + 1 \right] \e^2+ O(\e^3).
    \end{equation}
 $\theta$   is also the lowest eigenvalue of $L_b^{\textrm{resc}}$ on $(\Gamma_b^c)^{\perp}$ and therefore  $b \theta$  is the lowest eigenvalue of $L_b$ on $\G_b^{\perp}$.
    \end{proposition}
\begin{proof}
We use the expansions above to  expand $K_b$ in powers of $\e$ and the relation $\partial_{a^0}^*\psi^0 = 0$, to simplify the resulting terms to obtain $K_b = K^0 + \e W^1 + \e^2 W^2 + o(\e^3)$, where
\begin{equation}\label{K0}
	K^0 = \left( \begin{array}{cccc}
		-\COVLAP{a^0} - 1 & 0 & 0 & 0 \\
		0 & -\overline{\COVLAP{a^0}} - 1 & 0 & 0 \\
		0 & 0 & -\Delta & 0 \\
		0 & 0 & 0 & -\Delta
		\end{array} \right),
\end{equation}
\begin{equation} \label{W0}
	W^1
	= \left( \begin{array}{cccc}
		0 & 0 & 0 & i(\partial_{a^0}\psi^0) \\
		0 & 0 & -i(\overline{\partial_{a^0}\psi^0}) & 0 \\
		0 & i(\partial_{a^0}\psi^0) & 0 & 0 \\
		-i(\overline{\partial_{a^0}\psi^0}) & 0 & 0 & 0
		\end{array} \right),
\end{equation}
\begin{equation}\label{W1}
	W^{2} = \left( \begin{array}{cccc}
			 B^0
				& \kappa^2(\psi^0)^2 & 0 & 0 \\
			\kappa^2(\bar{\psi}^0)^2
				& \overline{B^0} & 0 & 0 \\
			0 & 0 & |\psi^0|^2 & 0 \\
			0 & 0 & 0 & |\psi^0|^2
		\end{array}	\right),
\end{equation}
where
\begin{equation}
	B^0 =  -\lambda^1 + 2\kappa^2|\psi^0|^2 -ia^1 \partial^*_{a^0}+ i\bar{a}^1 \partial_{a^0}.
\end{equation}


The unperturbed operator $K^0$ reduces to the operators studied previously, e.g. in \cite{TS}, where these operators are denoted by  $L^n$ with $n = 1$ and $M$ and where the spectra of the latter operators are described in details. In particular it is shown there that, in general, $K^0$ has the eigenvectors $(\psi^0, \bar{\psi}^0, 0, 0)$, $(\psi^0, \bar{\psi}^0, 0, 0)$, $(0,0,1,0)$ and $(0,0,0,1)$. The latter two are ruled out by the condition that $\lan \al \ran = 0$. We summarize the properties of $K^0$ in the following lemma.
\begin{lemma}
    $K^0$ is a nonnegative self-adjoint operator with discrete spectrum. It has a zero eigenvalue of multiplicity $2$ and the kernel is spanned by the elements
    \begin{equation} \label{v0rho}
        v_{0\rho} = (\psi^0,\rho\bar{\psi}^0,0,0), \ \quad \rho = \pm.
    \end{equation}
\end{lemma}
 Note that $v_{0\pm}\in V_{\pm}$. Hence the operator $K^{0}_+:=K^{0}|_{V_+}$  which is the zero-order approximation of the operator $K_+:=K_{b}|_{V_+}$ that we are interested in, has the simple lowest eigenvalue $0$ with the eigenfuction $v_{0+}$.
Also, the zero mode $v_{0-}$ is related to the complexified gauge zero mode $\G_b^c := \pi_- \G_b =(i\psi_b, -i\bar{\psi}_b, 0, 0)$.  Indeed, expanding the latter vector in $\e$, we obtain
\[\G_b^c = i \e  v_{0 -} + O(\e^3).\] 

\DETAILS{Now suppose that $v \perp \Gamma$. It follows that $\lan v, v_{0-} \ran_{L^2} = O(\e^2)$. We write $v = w + cv_{0-}$, with $w \perp v_{0-}$, and $c = O(\e^2)$. We then have, using $K_0 v_{0-} = 0$, that
\begin{align}
    \lan v, K_b v \ran
    &= \lan w, K_b w \ran + c \lan w, K_b v_{0-} \ran + \bar{c} \lan v_{0-}, K_b w \ran + |a|^2 \lan v_{0-}, K_b v_{0-} \ran \\
    &= \lan w, K_b w \ran + O(\e^3).
\end{align}
It therefore suffices to consider the restriction of the operator $K_b$ to the (closed) subspace consisting of elements orthogonal to $v_{0 -}$.}

By standard perturbation theory (see e.g. \cite{RSIV, HS, GS}), the spectrum of $K_{b}$ consists of eigenvalues which cluster in  $\e-$neighbourhoods of the eigenvalues of $K^{0}$ and each cluster has the same total multiplicity as the eigenvalue of $K^{0}$ it originates from. Of course, the same is true for its restriction $K_{+}:=K_b|_{V_+}$. Namely, the spectrum of $K_{+}$ consists of eigenvalues which cluster in  $\e-$neighbourhoods of the eigenvalues of $K^{0}_+:=K^{0}|_{V_+}$ and each cluster has the same total multiplicity as the corresponding eigenvalue of $K^{0}_+$. Thus $K_{+}$ has the simple eigenvector $v_{b}$, which is a perturbation of the simple eigenvector $v_{0+}$ of $K^{0}_{+}$ with the smallest eigenvalue $0$. It suffices to determine the corresponding eigenvalue of $K_{+}$, which we denote $\theta$. Since $K_{b}$ is self-adjoint and $\G_b^c$ is its eigenfunction with eigenvalue $0$, we have that $v_{b}\perp \G$. (In the leading order this becomes $v_{0+}\perp v_{0-}$ which can be verified directly.) Hence $\theta$ is the smallest eigenvalue of $K_{+}$ on the subspace orthogonal to $\G_b^c$.
%

To find $\theta$ we use the Feshbach-Schur map argument (see e.g. \cite{GS2, BFS}) with the projection $P$ given by the orthogonal projection onto $v_{0+}\in \NULL K^0_{+}$. This argument implies that \begin{equation} \label{isosp}
\lambda \in \sigma(K_{+})\ \quad \textrm{if and only if}\ \quad  \lambda \in \sigma(\mathcal{F}_P(\lambda)),
\end{equation} where, with $\bar P = 1 - P$,
\begin{equation} \label{fesh}
	\mathcal{F}_P(\lambda) := \left[P K_{+} P
		- PK_{+}\bar P( \bar P K_{+}\bar P -\lambda)^{-1}\bar P K_{+} P\right]_{\Ran P},
\end{equation}
provided the operator $\bar P K_{+}\bar P -\lambda$ is invertible on $\Ran \bar P$ and $\bar P K_{+}\bar P$ is bounded. (The latter conditions suffice for the right hand side of \eqref{fesh} to be well defined.)
Due to the relation  $\bar P K_{b}\bar P=\bar P (K_{+}- K^{0}_+)\bar P$ and the straightforward estimate
\begin{equation}  \label{estW}
	\|K_{+}- K^{0}_+\| \lesssim \e,
\end{equation}
we see that the operator $\bar P K_+\bar P$ is bounded. To show the invertibility of the operator $\bar P K_{+}\bar P -\lambda$ on $\Ran \bar P$, we note that it is the restriction of the operator $\bar Q K_+\bar Q-\lam$, where  $Q$ is the orthogonal projection onto $\NULL K^0$ and $\bar Q:=\one-Q$, to the subspace $V_+$   We know that $\sigma (\bar Q K^{0} \bar Q) \subset [\nu_0, \infty)$ for some $\nu_0 > 0$ and therefore, by standard perturbation theory 
we have that
\begin{equation} \label{specLbar}
	\sigma (\bar Q K_{b} \bar Q |_{\Ran \bar Q})\subset [c, \infty),
\end{equation}
with $c = \nu_0 + O(\e)$. Hence the self-adjoint operator $\bar Q K_{b}\bar Q -\lambda$ is invertible on $\Ran \bar Q$, provided $\lam <c$, and therefore its restriction $\bar P K_{+}\bar P -\lambda$ (to real-linear subspace $V_+$) is invertible on $\Ran \bar P$ and $\|(\bar P K_{+} \bar P - \lambda)^{-1}\| \le c^{-1}$ (again provided $\lam <c$). Hence \eqref{fesh} is well defined for $\lam <c$.

We now use $K_{+} = K^0_{+} + \e W^1  + \e^2 W^2+ o(\e^3)$, the relation $K^0_{+} P = P K^0_{+} = 0$ and the facts   $\|P K_{+} \bar P \| =O(\e)$ (by \eqref{estW}) and $\|(\bar P K_{+} \bar P - \lambda)^{-1}\| \lesssim 1$, provided $\lambda < c$ (by \eqref{specLbar}). Since we are studying the eigenvalue in $O(\e)-$ neighbourhood of $0$, we have that $\lambda = O(\e)$. Using this, we obtain
\begin{equation}\label{Fexp}
	\mathcal{F}(\lambda) = \e\mathcal{F}_1
		+ \e^2 \mathcal{F}_2+ O(\e^3),
\end{equation}
where
\begin{equation}\label{F1F2} \mathcal{F}_1 :=\lan v_{0+}, W^1 v_{0+}\ran/\lan |v_{0+}|^2 \ran_\cL,\ \mathcal{F}_2:= \lan v_{0+}, [W^2 -	 W^1\bar{P}(\bar{P}K^{0}_{+}\bar{P})^{-1}\bar{P}W^1]v_{0+}\ran/\lan |v_{0+}|^2 \ran_\cL.
\end{equation}
The operator, $W_1$ is explicitly given by \eqref{W0}. This expression implies that  $\lan v_{0+}, W_1 v_{0+} \ran = 0$ and therefore 
\begin{equation}\label{F1}
	\mathcal{F}_1 := \lan v_{0+}, W^1 v_{0+}\ran = 0.
\end{equation}	

We now turn to the $\e^2$ order operator, $\mathcal{F}_2$.  
\begin{lemma}\label{lem:F2}
    \begin{equation}    \label{F2}
        \cF_{2} 
	=2 [(2\kappa^2 - 1)\beta(\tau) + 1]. 
    \end{equation}
\end{lemma}
\begin{proof}
    We begin with $\lan v_{0+}, W^2 v_{0+}\ran$. We first note that $W^2$ and $v_{0\rho}$ are explicitly given by \eqref{W1} and \eqref{v0rho}. Using the fact that $\partial_{a^0}^*\psi^0 = 0$, we calculate that
    \begin{equation}    \label{F2firstterm}
        \lan v_{0\rho}, W^2  v_{0\rho}\ran =2 \left( -\lambda_1\lan |\psi^0|^2\ran_\cL + 2\kappa^2\lan|\psi^0|^4\ran_\cL +\Re\lan\bar\psi^0 i\bar{a}^1\partial_{a^0}\psi^0\ran_\cL+\rho\kappa^2\lan|\psi^0|^4\ran_\cL \right),
	\end{equation}
where $\lam_1$ is defined by the expansion $\lambda_b = 1 + \e^2 \lambda^1 + O(\e^4)$ and, due to the definition of $\lam_b$ and $\e$, is equal to
 \begin{equation} \label{relationaver}\lambda^1=\left[\frac{1}{2}  + (\kappa^2-{1\over 2}) \beta(\tau)\right]\langle |\psi^{0} |^2 \rangle_\cL. \end{equation}
    We note that $-\Delta = \bar{\partial}^* \bar{\partial}$. Using the identities \eqref{a1eq} and \eqref{curla1} we obtain
    \begin{equation}    \label{F2identity}
        2\Re \lan i\bar{a}^1 \bar{\psi}^{0} \partial_{a^0} \psi^{0} \ran_\cL
        = 4\lan \bar{a}^1\Delta a^1 \ran_\cL
        = -4\lan |\bar{\partial} a^1|^2 \ran_\cL
        = - \lan |\psi^0|^4 \ran_\cL + \lan |\psi^0|^2 \ran_\cL^2.
	\end{equation}
    The equations \eqref{F2firstterm} and \eqref{F2identity}
     and the relation \eqref{relationaver} 
     give 
	\begin{equation}  \label{W2matrelem}
	   \lan v_{0\rho}, W^2 v_{0\rho} \ran = 2(1 + \rho)\kappa^2 \lan |\psi^0|^4 \ran_\cL.
	\end{equation}

    To compute the second term in $\cF_2$ we note that $\bar{P}W^1 P = W^1 P$, and use \eqref{K0}, we calculate
    \begin{equation}    \label{F2second}
    \lan v_{0\rho}, W^{1} \bar{P} (\bar{P} K^{0 }_{+} \bar{P})^{-1} \bar{P} W^1 v_{0\rho} \ran
        = -2\Re \lan  \bar{\psi}^{0} (\partial_{a^0}\psi^0) \Delta^{-1} ( \psi^{0} (\overline{\partial_{a^0}\psi^0}) ) \ran_\cL.
    \end{equation}
    Again using the identities \eqref{a1eq} and \eqref{curla1} and the fact that  $-\Delta = \bar{\partial}^* \bar{\partial}$, we obtain
	\begin{equation}    \label{F2secondterm}
     \begin{split}   -2\Re \lan \bar{\psi}^{0} (\partial_{a^0}\psi^0) \Delta^{-1}
	 			( \psi^{0} (\overline{\partial_{a^0}\psi^0}) ) \ran_\cL
        &= -8 \Re \lan \bar{a}^1 \Delta a^1 \ran_\cL= 8 \lan |\bar{\partial} a^1|^2\ran_\cL\\
        &
        = 2\lan |\psi^0|^4 \ran_\cL - 2\lan |\psi^0|^2 \ran_\cL^2.
\end{split}	\end{equation}
The second equation in \eqref{c0vac} and the fact that $\lan |v_{0\rho}|^2 \ran_\cL = 2\lan |\psi^0|^2 \ran_\cL$ now gives \eqref{F2}.
\end{proof}

The equations \eqref{isosp}, \eqref{Fexp}, \eqref{F1} and \eqref{F2} imply the first part of Preposition \ref{prop:qfKbnd}.
Finally, 
 by the definition, $\theta$   is also the lowest eigenvalue of $L_b^{\textrm{resc}}$ on $\G_1^{\perp}$ and therefore, by the formula relating $L_b^{\textrm{resc}}$ to $L_b$, we see that $\theta b$ is the smallest eigenvalue of $L_b$ on $\G_b^{\perp}$.
\end{proof}

Now we are ready for
\begin{proof}[Proof of Theorem \ref{thm:positivity}]
We restore the subindex $L^2$ in the inner products and norms. Preposition \ref{prop:qfKbnd} implies the first statement of the theorem and the estimate
\begin{equation}\label{qfKbnd}
        \lan v, K_b v \ran_{L^2} \geq  \theta  \|v\|_{L^2}^2.
    \end{equation}
  for all $v \in D(K_b)$, such that $v \perp \Gamma_b^c$. The latter estimate gives,  for all $v \in D(L_b)$, such that $v \perp \Gamma_b$,
\begin{equation}\label{qfLbnd}
        \lan v, L_b v \ran_{L^2} \geq \frac{1}{2}\theta  \|v\|_{L^2}^2.
           \end{equation}
    We upgrade now the lower bound on $\lan v, L_b v \ran_{L^2}$ to that on $\lan v, L_b v \ran_{H^1}$. (We could have done with the operator $K_b$ as well.) We begin with
\begin{lemma}\label{lem:H1-bound}
    For all $v \in \Per_b$, if $v \perp \Gamma_b$, then
    \begin{equation}    \label{eq:H1-bound}
\frac{1}{2} \|v\|_{H^1}^2 - C \|v\|_{L^2}^2 \le  \langle v, L_b v \rangle_{L^2} \lesssim \|v\|_{H^1}^2,
    \end{equation}
    for some positive constant $C$.
\end{lemma}
\begin{proof}
    We write $v = (\xi, \alpha)$ and also $A^\tau_b = A^0_b + P$. For convenience we simplify the notation in the following calculations. Integrating by parts we obtain
    \begin{align}\label{Lb-qf}
\langle v, L_b v \rangle_{L^2}    \nonumber
&= \frac{1}{|\Omega|} \Re \int_\Omega -\bar{\xi}\COVLAP{A^0_b}\xi + 2iP \cdot \bar{\xi}\COVGRAD{A^0_b}\xi + |P|^2|\xi|^2 - \kappa^2|\xi|^2 + (2\kappa^2 + \frac{1}{2})|\Psi_b|^2|\xi|^2 \\ \nonumber
            &  + (\kappa^2 - \frac{1}{2})\Psi_b^2\bar{\xi}^2
            + 2i\alpha\cdot(\bar{\xi}\COVGRAD{A^0_b}\Psi_b) + 2(\alpha\cdot P)\Psi_b\bar{\xi}
            - \alpha \cdot \Delta \alpha + |\Psi_b|^2|\alpha|^2 \\ \nonumber
            &    - 2 \alpha \cdot \Im(\bar{\xi}\COVGRAD{A^0_b}\Psi) - 2(\alpha\cdot P)\Re(\bar{\xi}\Psi_b) \\ \nonumber
        &= \frac{1}{|\Omega|} \int_\Omega |\COVGRAD{A^0_b}\xi|^2 + |\nabla \alpha|^2 + |P|^2|\xi|^2 - \kappa^2|\xi|^2 + (2\kappa^2 + \frac{1}{2})|\Psi_b|^2|\xi|^2 \\
            & + (\kappa^2 - \frac{1}{2})\Psi_b^2\bar{\xi}^2
            + |\Psi_b|^2|\alpha|^2 - 2P \cdot \Im(\bar{\xi}\COVGRAD{A^0_b}\xi)
            - 4 \alpha \cdot \Im(\bar{\xi}\COVGRAD{A^0_b}\Psi_b).
   \end{align}
     Using this expression and the estimate obtained with the help of the Schwarz inequality
        \begin{align}\label{cross-term-est}
\left| \int_\Omega P \cdot \Im(\bar{\xi}\COVGRAD{A^0_b}\xi) \right|
&\leq \|P\|_\infty \|\xi\|_{L^2} \|\COVGRAD{A^0_b}\xi\|_{L^2}
\leq \frac{1}{2} \|P\|_\infty \left( \frac{1}{r} \|\xi\|_{L^2}^2 + r \|\COVGRAD{A^0_b}\xi\|_{L^2}^2 \right),
    \end{align}
for any $r>0$, we obtain
  \begin{align*}
        \langle v, L_b v \rangle_{L^2}
                &\ge \|\COVGRAD{A^0_b}\xi\|_{L^2}^2 + \|\nabla \alpha\|_{L^2}^2
            - \kappa^2\|\xi\|_{L^2}^2 \\  &
             - \frac{C}{r} \|\xi\|_{L^2}^2 - r C \|\COVGRAD{A^0_b}\xi\|_{L^2}^2
            - C\|\alpha \|_{L^2}^2- C\|\xi\|_{L^2}^2.
    \end{align*}
  \DETAILS{    This means that
    \begin{align*}
        \int_\Omega -2P \cdot \Im(\bar{\xi}\COVGRAD{A^0_b}\xi)
        \geq -\frac{C}{r^2} \|v\|_{L^2}^2 - r^2 C \|v\|_{H^1}^2.
    \end{align*}}
  Now we choose $r$ as $r=\frac{1}{2C}$ so that we arrive at the lower bound in  \eqref{eq:H1-bound}. To obtain the upper bound we use \eqref{Lb-qf} and \eqref{cross-term-est} again and the fact that $\|\Psi_b\|_\infty,\ \|P\|_\infty< \infty$. 
\end{proof}
    Let now $\delta \in [0,1]$ be arbitrary. Then using \eqref{qfKbnd} 
    and \eqref{eq:H1-bound}
    \begin{align*}
        \lan v, L_b v \ran_{L^2}
        &= (1 - \delta) \lan v, L_b v \ran_{L^2} + \delta \lan v, L_b v \ran_{L^2} \\
        &\geq (1 - \delta) \theta \|v\|_{L^2}^2 + \delta (\frac{1}{2}\|v\|_{H^1}^2 - C \|v\|_{L^2}^2) \\
        &= ((1 - \delta) \theta - \delta C ) \|v\|_{L^2}^2 + \frac{\delta}{2}  \|v\|_{H^1}^2.
    \end{align*}
    \eqref{L-low-bound} now follows by choosing $\delta = \frac{\theta}{\frac{1}{2} +\theta + C}$.
Next, the 
estimate \eqref{L-upp-bound} follows from the upper bound in  \eqref{eq:H1-bound}.    
\end{proof}


\appendix

\section{Proof of orbital stability}

We now assume that we have the case $\kappa^2 > \frac{1}{2}(1 - \frac{1}{\beta(\tau)})$ and prove the orbital stability of the Abrikosov lattice $u^\tau_b$. This is a weaker statement than the asymptotic stability which we have already proven but it requires rougher analysis. We follow \cite{G}.
\begin{theorem}\label{thm:orb-stability}
    For all $b$ sufficiently close to $\kappa^2$, the Abrikosov lattice $u^\tau_b$ is orbitally stable under gauge-periodic perturbations if $\kappa^2 > \frac{1}{2}(1 - \frac{1}{\beta(\tau)})$.
\end{theorem}
\begin{proof}  As in the main text, we consider $\tau$ fixed and do not display it in the notation.
We will require a series of lemmas.
\begin{lemma}
    There exists positive constants $c$ and $C$, such that for all $v \in \Per_b$, if $v \perp \Gamma_b$, then for any $\theta'<\theta$,
    \begin{equation}    \label{lyap-bound}
    \theta'\|v\|_{H^1}^2 - c\|v\|_{H^1}^3 - c\|v\|_{H^1}^4 \leq \E(u_b + v) - \E(u_b) \leq C(\|v\|_{H^1}^2 + \|v\|_{H^1}^3 + \|v\|_{H^1}^4).
    \end{equation}
\end{lemma}
\begin{proof}
    Using the Taylor expansion of $\E$, together with the fact that $\E'(u_b) = 0$ and $v \perp \Gamma_b$, we have
    \begin{align*}
        \E(u_b + v) - \E(u_b) = \frac{1}{2} \lan v, L_b v \ran + R(v),
    \end{align*}
    where the remainder $R(v)$ is given by, setting $v = (\xi, \alpha)$,
    \begin{align*}
        R(v) = \int_{\Omega_b} |\alpha|^2(\Re(\bar{\Psi}_b \xi) + \frac{1}{2} |\xi|^2) - \alpha\cdot\Im(\bar{\xi}\COVGRAD{A_b}\xi) + \kappa^2|\xi|^2(\Re(\bar{\Psi}_b \xi) + \frac{1}{4} |\xi|^2).
    \end{align*}
    Using the Cauchy-Schwarz and Sobolev inequalities it is straightforward to show $|R(v)| \leq c(\|v\|_{H^1}^3 + \|v\|_{H^1}^4)$. Using \eqref{L-low-bound}, this gives
    \begin{align*}
        \E(u_b + v) - \E(u_b) \geq \theta\|v\|_{H^1}^2 - c\|v\|_{H^1}^3 - c\|v\|_{H^1}^4.
    \end{align*}

    On the other hand, by the definition \eqref{L} of $L_b$ and the boundedness of $u^\tau_b$ together with its derivatives, we have $\lan v, L_b v \ran \lesssim \|v\|_{H^1}^2$. This estimate together with the above estimate of $|R(v)|$ gives the upper bound in \eqref{lyap-bound}
        and this completes the proof.
\end{proof}

\begin{lemma}\label{lem:decreasing-energy}
    Suppose that $(\Psi, A, \Phi)$ is a solution of the Gorkov-Eliashberg-Schmidt equations \eqref{GES} on the time interval $[0,T]$ satisfying (I) - (III) and $u(t) = (\Psi(t), A(t)) \in C^1([0,T]; U_\delta)$ for any $\delta > 0$ and $T > 0$. Then the energy function $\E(u)$ is a nonincreasing function in time.
\end{lemma}
\begin{proof}
    Note that we have $\partial_t u \in \Per_b$. Using the gradient-flow form of the equations \eqref{grad-flow}, we see that
    \begin{align*}
        \partial_t \E(u)
        &= \lan \E'(u), \partial_t u \ran_{L^2} \\
        &= - \lan \E'(u), (i\Phi\Psi, \nabla\Phi) \ran_{L^2} - \| \E'(u) \|_{L^2}^2 \\
        &\leq \int_{\Omega_b} \Im(\Phi\bar{\Psi}\COVLAP{A}\Psi) - \Re(i\kappa^2(1 - |\Psi|^2)|\Psi|^2\Phi) - \CURL^*\CURL A \cdot \nabla\Phi + \Im(\bar{\Psi}\COVGRAD{A}\Psi)\cdot\nabla\Phi \\
        &= \int_{\Omega_b} -\Im(\bar{\Psi}\COVGRAD{A}\Psi)\cdot\nabla\Phi - \Im(\Phi|\COVGRAD{A}\Psi|^2) + \Im(\bar{\Psi}\COVGRAD{A}\Psi)\cdot\nabla\Phi \\
        &= 0,
    \end{align*}
    where we used the fact that $\Phi$ is real-valued and $\DIV \CURL^* = 0$.
\end{proof}

\begin{lemma}\label{lem:v-control}
Let  $(\Psi(t), A(t), \Phi(t))$ be a solution of the Gorkov-Eliashberg-Schmidt equations \eqref{GES} on the time interval $[0,T]$ and denote $u(t) = (\Psi(t), A(t))$.    Given $\e > 0$ sufficiently small, there exists $\delta > 0$ such that if $T_{\gamma(t)}u(t) = u_b + v(t)$, $\gamma(t) \in \R$, $v \in C^1([0,T]; \Per_b)$, $v(t) \perp \Gamma_b$, and $\|v(0)\|_{H^1} < \delta$, then $\|v(t)\|_{H^1} \leq \e$ for all $t \in [0,T]$.
\end{lemma}
\begin{proof}
    We set $N(t) = \|v(t)\|_{H^1}$. Using the inequalities \eqref{lyap-bound} and Lemma \ref{lem:decreasing-energy} we have
    \begin{align*}
    C_1 N(t)^2 - C_2 N(t)^3 - C_3 N(t)^4
    &\leq \E(u_b + v(t)) - \E(u^\tau_b)= \E(T_{\gamma(t)} u(t)) - \E(u_b) \\
    &= \E(u(t)) - \E(u_b)\leq \E(u(0)) - \E(u_b) \\
    &= \E(T_{\gamma(0)}u(0)) - \E(u_b)= \E(u_b + v(0)) - \E(u_b) \\
    &\leq  C_4 N(0)^2 + C_5 N(0)^3 + C_6 N(0)^4.
    \end{align*}
    Now there exists $\delta_0 > 0$ such that if the left-hand-side $C_1 N(t)^2 - C_2 N(t)^3 - C_3 N(t)^4 \leq \delta_0$, then either $0 \leq N(t) \leq \e$ or $N(t) \geq \e'$ for some $e' > \e$. We can choose $\delta$ sufficiently small so that the right-hand-side $C_4 N(0)^2 + C_5 N(0)^3 + C_6 N(0)^4 \leq \delta_0$. The result then follows from the continuity of $N(t)$.
\end{proof}

We can now prove the following proposition, which implies Theorem \ref{thm:orb-stability}.
\begin{proposition}
    For any $\e > 0$, there exists $\delta > 0$ such that if $(\Psi, A, \Phi)$ is a $C^1$ solution of the Gorkov-Eliashberg-Schmidt equations \eqref{GES} and $u(t) = (\Psi(t), A(t))$ satisfies $u(0) \in U_\delta$, then $u(t) \in U_\e$ for all $t \geq 0$.
\end{proposition}
\begin{proof}
    Let $\delta_0$ be given by Proposition \ref{prop:decomposition}. If $u(0) \in U_{\frac{1}{2}\delta_0}$, there is $T > 0$ such that $u(t) \in U_{\delta_0}$ for $t\le T$. Then by Proposition \ref{prop:decomposition}, there is $\gamma(t) \in \R$ so that $T_{\gamma(t)}u(t) = u_b + v(t)$, with $v(t) \in C^1([0,T]; \Per_b)$ and $v(t) \perp \Gamma_b$. By Lemma \ref{lem:v-control} we can then find $\delta_1$ such that $\|v(t)\| \leq \frac{1}{2}\min(\e, \delta_0)$ for all $t \in [0,T]$.

    Now let $\delta = \min(\delta_0, \delta_1)$. It follows that if $u(0) \in U_{\frac{1}{2}\delta}$, then $u(t)$ cannot leave $U_\e \cap U_{\delta_0}$, and that proves the proposition.
\end{proof}
This finishes the proof of  Theorem \ref{thm:orb-stability}.
\end{proof}
\bigskip



\end{document}